\documentclass[copyright,creativecommons]{eptcs}
\providecommand{\event}{AFL 2026} 

\usepackage{iftex}

\ifpdf
  \usepackage{underscore}         
  \usepackage[T1]{fontenc}        
\else
  \usepackage{breakurl}           
\fi
\usepackage{amsfonts}
\usepackage{amssymb}
\usepackage{amsmath}
\usepackage{amsthm}
\usepackage{bbm}
\usepackage{mathtools}
\usepackage{stmaryrd}
\usepackage{hyperref}
\usepackage{enumitem}
\usepackage{xifthen}

\usepackage{xcolor}
\usepackage{microtype}
\AtBeginDocument{\microtypesetup{protrusion=false,expansion=false}}

\usepackage{bbding}
\DeclareMathOperator{\Ima}{Im}
\DeclareMathOperator{\rank}{rk}
\DeclareMathOperator{\mxrk}{mxrk}
\newcommand*{\SET}[1]{\{#1\}}
\newcommand*{\SETC}[2]{\{#1\,|\,#2\}}
\newcommand{\nat}{\mathbb N}
\newcommand{\aut}{\mathcal A}

\newcommand{\sring}{\mathbb K}
\newcommand{\sringset}{K}
\newcommand{\bigO}{O}
\newcommand{\bigOmega}{\Omega}
\newcommand{\sig}{\Sigma}
\newcommand{\trees}{\mathrm T}
\newcommand{\subst}[2]{#1\llbracket#2\rrbracket}
\DeclareMathOperator{\wt}{wt}
\newcommand{\lang}[1][]{\ifthenelse{\isempty{#1}}{L}{L_{#1}}}%
\newcommand{\polynom}[2]{\mathrm{Pol}(#1, #2)}

\newcommand{\X}{\mathrm X}
\newcommand{\Z}{\mathrm Z}
\newcommand{\Rec}[1]{\mathrm{Rec}(#1)}
\newcommand{\RecA}[1]{\mathrm{Rec_{alt}}(#1)}
\newcommand{\RecPol}[1]{\mathrm{Rec_{pol}}(#1)}
\newcommand{\HOM}{\xhom{}{}{}}
\newcommand{\xhom}[3]{\mathrm{\ifthenelse{\isempty{#1}}{}{#1\text{-}}HOM}\ifthenelse{\isempty{#2}}{}{_{#2}}\ifthenelse{\isempty{#3}}{}{(#3)}}

\def\claimname{Claim}

\def\definitionname{Definition}
\def\examplename{Example}

\def\lemmaname{Lemma}

\def\propositionname{Proposition}

\def\sectionname{Section}

\def\theoremname{Theorem}

\theoremstyle{plain}
\newtheorem{theorem}{Theorem}[section]

\newtheorem{corollary}[theorem]{Corollary}
\newtheorem{lemma}[theorem]{Lemma}

\newtheorem{proposition}[theorem]{Proposition}

\theoremstyle{definition}
\newtheorem{definition}[theorem]{Definition}
\newtheorem{example}[theorem]{Example}

\theoremstyle{remark}
\newtheorem{remark}[theorem]{Remark}

\newtheorem{claim}{Claim}[theorem]

\title{Weighted Alternating Tree Automata}
\author{Olle Torstensson
\institute{Link\"{o}ping University\\ Link\"{o}ping, Sweden}
\email{olle.torstensson@liu.se}
}
\def\titlerunning{Weighted Alternating Tree Automata}
\def\authorrunning{O. Torstensson}
\begin{document}
\maketitle

\begin{abstract}
We define weighted alternating tree automata --- a generalization of weighted (non-alternating) tree automata, (unweighted) alternating tree automata, and weighted alternating (string) automata --- over commutative semirings. We study their expressive power, and show that it is equivalent to that of weighted tree automata exactly when the employed semiring is locally finite. In our main result, we characterize the class of weighted tree languages defined by weighted alternating tree automata as the closure of the recognizable weighted tree languages under inverse tree homomorphisms.
\end{abstract}

\section{Introduction}
\label{sec:intro}

The basic model of finite automata over ordered, finite, and labeled trees has seen a great number of generalizations over time. One direction has been to equip the automaton with a weight structure, usually a semiring, allowing it to compute a quantitative output, resulting in a \emph{weighted tree automaton (wta)}.
The output of a wta is produced by letting the addition and multiplication operators of the semiring accumulate the individual transition weights (elements of the semiring) that are attached to each transition, according to a specific scheme. Alternation is an orthogonal generalization that introduces the concept of universality (roughly, \emph{all} computations going forward should be accepting) to the automaton, which in the basic version is only equipped to handle existentiality (roughly, \emph{some} computation going forward should be accepting).

In this paper, we perform both these generalizations at once by introducing \emph{weighted alternating tree automata (wata)} over commutative semirings. Intuitively, they differ from ordinary weighted tree automata in that they allow arbitrary polynomials over the semiring to describe how computations on subtrees are to be combined, as opposed to having these polynomials be of a specific limiting form; while wta have the power to non-deterministically perform multiple computations on the same subtree and sum up the results (corresponding to existentiality), wata may also \emph{multiply} them (corresponding to universality). For example, in a wata over, e.g., the natural semiring and some alphabet containing a $\gamma$ of positive rank there may be a transition
\begin{equation*}
    \delta(q, \gamma) = 2 \cdot x_{(q,1)}^2 \,.
\end{equation*}
where the variable $x_{(q,1)}$ represents the computed weight of the first direct subtree in the state $q$. The kind of non-linearity displayed here, where references to more than one computation carried out on the \emph{same subtree} occur \emph{in the same monomial}, allows for some complex weighting patterns generally not achievable by wta.

Tree automata were extended with alternation by Slutzki~\cite{slutzki1985alternating-tre}, who showed that, while it may add succinctness to the formalism, it adds nothing in terms of expressive power.
\sectionname~\ref{sec:two-mode} explicitly establishes wata as the weighted generalization of alternating tree automata, and we will show that, over locally finite semirings, this holds true also in the weighted case. In contrast, by leveraging a result on weighted alternating automata over strings, we will see that alternation does indeed contribute expressive power to weighted tree automata whenever the employed semiring is not locally finite (\sectionname~\ref{sec:power}).

In \sectionname~\ref{sec:characterization} we will see that there are intimate ties between the weighted version of universality described above, and the notion of copying subtrees. Intuitively, if each subtree is duplicated ahead of being processed, with each copy being treated as a unique subtree, computations on what is essentially the same subtree may be multiplied freely, and the same effect as that of the earlier described non-linearity in the transitions can be achieved using wta style transitions. Using this intuition, we can precisely capture the complexity added to wta by alternation via (non-linear) tree homomorphisms. Whereas the unweighted recognizable tree languages are closed under inverse tree homomorphisms without restrictions, the same does not hold in general in the weighted case. This is due precisely to the copying capability of non-linear tree homomorphisms, which allows for exponential expansions of trees and thereby weight computations that can be too complex (e.g., grow too fast) to be handled by a wta on the original tree. A wata, however, is adequately equipped to compute just such functions. In fact, we characterize the weighted tree languages defined by wata as the closure of the recognizable tree languages under inverse tree homomorphisms (\theoremname~\ref{thm:characterization}).

\paragraph{Related work.}
\label{sec:related}
Weighted alternating automata for \emph{strings} were introduced over commutative semirings by Kostol{\'a}nyi and Mi{\v s}{\'u}n~\cite{kostolanyi2018alternating-wei}, and then further studied by Grabolle~\cite{grabolle2023a-nivat-theorem}. Leveraging the fact that strings are just trees without branching, we use some of their results.

The weighted alternating parity tree automata of~\cite{baader2016reasoning-with-} are introduced as a means to extend Description Logics and are not given much treatment in terms of expressive power. Their notion of trees is different from ours in that they, e.g., are allowed to be infinite and consequentially their model does not have all that much in common with ours.

Ghorani and Zahedi~\cite{ghorani2016alternating-reg} define several variations of tree automata and tree grammars implementing alternation in a weighted setting, and compare them in terms of expressive power. While their most basic automaton model is defined in a way similar to ours, it employs a complete residuated lattice as a weight structure, making results largely incomparable.

Finally, an equivalent version of the automaton introduced in this paper has appeared as an ad-hoc device in the larger context of dynamically weighted tree transducers~\cite{drewes2025dynamically-wei}. Apart from its connection to top--down weighted tree transducers with regular look-ahead, it has not been studied in any depth, and not explicitly connected to the theory of alternation.

\section{Preliminaries}
\label{sec:prelim}

We let $\nat = \SET{0, 1,\dots}$ denote the natural numbers, and for each $k \in \nat$, we let $[k] = \SET{1, \dots, k}$. The inverse of a function $f$ is denoted $f^{-1}$, and its image $\Ima(f)$. For any two functions $f_1\colon A \to B$ and $f_2\colon B \to C$, we let $f_1 \mathrel{;} f_2$ denote their composition, i.e., $(f_1 \mathrel{;} f_2)(a) = f_2(f_1(a))$ for every $a \in A$. We define the composition of classes $F_1$, $F_2$ of functions similarly: $F_1 \mathrel{;} F_2 = \SETC{f_1 \mathrel{;} f_2}{f_1 \in F_1, f_2 \in F_2}$.
Furthermore, for any set $A$, we fix a set $\X_A = \SETC{x_a}{a \in A}$ of variables ranging over polynomials and a set $\Z_A = \SETC{z_a}{a \in A}$ of variables ranging over trees.\footnote{These variables are assumed to not appear a priori in any of the ranked alphabets considered in the paper.}

\paragraph{Trees.}
A \emph{ranked alphabet} is a finite non-empty set $\sig$ of \emph{symbols}, together with an associated function ${\rank_\sig \colon \sig \to \nat}$, such that $\rank_\sig^{-1}(0) \neq \emptyset$, giving each symbol its \emph{rank}.
For each $k \in \nat$, $\sig^{(k)}$ is the set of all symbols in $\sig$ of rank $k$, and the notation $\sigma^{(k)}$ is sometimes used to indicate that $\sigma \in \sig^{(k)}$. The \emph{maximal rank} of a ranked alphabet $\sig$, $\mxrk(\sig)$, is the maximal $k \in \nat$ such that $\sig^{(k)} \neq \emptyset$.
\begin{quote}
    \emph{For the rest of this paper, $\sig$ denotes an arbitrary ranked alphabet.}
\end{quote}
Given an arbitrary set $A$, the \emph{set of trees over $\sig$ indexed by $A$}, denoted $\trees_\sig(A)$, is the smallest set $T$ containing $A$ and $\sig^{(0)}$, and whenever $t_1, \dots, t_k \in T$, we have $\sigma[t_1, \dots, t_k] \in T$, for each $k \geq 1$ and $\sigma \in \sig^{(k)}$. We write $\trees_\sig$ in place of $\trees_\sig(\emptyset)$.
If $\mxrk(\sig) = 1$ and $| \sig^{(0)} | = 1$, we call $\sig$ a \emph{string alphabet} and denote the trees over $\sig$ without brackets.
A tree $t$ in $\trees_\sig(A)$ is called \emph{linear} (respectively, \emph{non-deleting}) with respect to $A' \subseteq A$ if each element in $A'$ appears in $t$ at most (respectively, at least) once. For a set $A$, a finite $I \subseteq A$, together with trees $t$ and $\SET{t_i}_{i \in I}$ in $\trees_\sig(\Z_A)$, we denote by $\subst{t}{z_i \leftarrow t_i \mid i \in I}$ the tree in $\trees_\sig(\Z_A)$ that is the result of simultaneously substituting each $z_{i}$ in $t$ with $t_i$, for all $i \in I$.

\paragraph{Semirings and polynomials.}
A \emph{commutative semiring} is a $5$-tuple $\sring = (\sringset, +, \cdot, 0, 1)$, consisting of a carrier set $\sringset$, two associative and commutative operators $+$ (addition) and $\cdot$ (multiplication), and two elements $0,1 \in \sringset$, such that: $0 + a = 1 \cdot a = a$, $0 \cdot a = 0$, and $a \cdot (b + c) = a \cdot b + a \cdot c$ for all $a,b,c \in \sringset$.
Addition and multiplication are extended to operate on finite sets via the standard symbols $\sum$ and $\prod$ (with an empty sum and product being equal to $0$ and $1$, respectively), and $a^d$ is as usual a shorthand for $\prod_{i \in [d]} a$, for any $a \in K$ and $d \in \nat$.
We will generally allow a symbol to denote both a semiring and its carrier set.
For any semiring $\sring$ and $A \subseteq \sring$, the \emph{subsemiring of $\sring$ generated by $A$} is the semiring $\langle A \rangle = (A', +, \cdot, 0, 1)$, where $A'$ is the closure of $A \cup \SET{0,1}$ under the addition and multiplication of $\sring$.
If $\langle A \rangle$ is finite (has a finite carrier set) whenever $A$ is finite, we call $\sring$ \emph{locally finite}. The \emph{boolean semiring} $\mathbb{B} = (\SET{0,1}, \vee, \wedge, 0, 1)$ is an example of a (locally) finite commutative semiring, and the \emph{natural semiring} $\nat = (\nat, +, \cdot, 0, 1)$ is an example of a commutative semiring that is \emph{not} locally finite.
\begin{quote}
    \emph{For the rest of this paper, $(\sring, +, \cdot, 0, 1)$ denotes an arbitrary commutative semiring.}
\end{quote}
Given a set $A$, we let $\polynom{\sring}{\X_A}$ denote the set of \emph{polynomials over variables in $\X_A$ and coefficients in $\sring$}. We will frequently and tacitly assume that polynomials in $\polynom{\sring}{\X_A}$ are uniquely represented as finite sums of distinct \emph{monomials} $m_1, \dots, m_\ell$ --- i.e., polynomials without addition --- with each $m_i$ being of the form $c_i \cdot \prod_{a \in A'_i} x_a^{d_{(i,a)}}$ for some finite $A'_i \subseteq A$, where $c_i \in \sring$ is the \emph{coefficient} of $m_i$ and the maximal $d_{(i,a)} \in \nat$ is its \emph{degree}. The maximal degree of any constituting monomial is then the degree of the polynomial.
Analogously to substitution on trees, given a set $A$, some finite $I \subseteq A$, together with polynomials $\pi$ and $\SET{\pi_i}_{i \in I}$ from $\polynom{\sring}{\X_{A}}$, we denote by $\subst{\pi}{x_{i} \leftarrow \pi_i \mid i \in I}$ the polynomial in $\polynom{\sring}{\X_A}$ in which every $x_i$ in $\pi$ has been substituted for $\pi_i$, for all $i \in I$. Along with substitution, $\polynom{\sring}{\X_A}$ is also closed under addition and multiplication --- it is in fact itself a commutative semiring under the same operations as those for $\sring$, extended to polynomials in the natural way~\cite{fulop2022weighted-tree-a}.

\paragraph{Weighted tree languages.}
A \emph{weighted tree language (over $\sig$ and $\sring$)} is a function $\lang \colon \trees_\sig \to \sring$. A \emph{weighted tree automaton (wta, or $(\sig, \sring)$-wta)} is a tuple $(\sig, Q, Q_d, R, \sring, \wt)$, consisting of a ranked alphabet $\sig$ (of input symbols), a finite non-empty set $Q$ (of states) with $Q \cap \sig = \emptyset$, a set $Q_d \subseteq Q$ (of designated states), a ranked alphabet $R$ (of transition rules) such that $R^{(k)} \subseteq Q \times \sig^{(k)} \times Q^k$ for all $k \geq 0$, a semiring $\sring$, and a function $\wt \colon R \to \sring$. For convenience, we often represent the third component of a transition rule as a string, and sometimes represent a rule $\rho = (q,\sigma, q_1\dots q_k)$ together with its weight as $q[\sigma] \xrightarrow{\wt(\rho)} \langle q_1, \dots, q_k \rangle$. For any $q \in Q$ and $\sigma \in \sig$, $R_{q;\sigma}$ is the set of all rules in $R$ of the form $(q, \sigma, w)$.

A wta $\aut = (\sig, Q, Q_d, R, \sring, \wt)$ induces, for each of its states $q \in Q$, a weighted tree language $\lang[\aut; q]$, defined for each $t = \sigma[t_1, \dots, t_k] \in \trees_\sig$ by
\begin{equation*}
    \lang[\aut; q](t) = \sum_{\rho = (q, \sigma, q_1\dots q_k) \in R} \wt(\rho) \cdot \prod_{i \in [k]} \lang[\aut; q_i](t_i) \, .
\end{equation*}
The \emph{language recognized by $\aut$} is the function $\lang[\aut] \colon \trees_\sig \to \sring$ such that $\lang[\aut](t) = \sum_{q \in Q_d} \lang[\aut; q](t)$ for each $t \in \trees_\sig$. Any weighted tree language recognized by a wta is called \emph{recognizable}, and we denote the set of all recognizable weighted tree languages (i) over $\sig$ and $\sring$ by $\Rec{\sig, \sring}$; (ii) over (some $\sig$ and) $\sring$ by $\Rec{\sring}$; and (iii) over some $\sig$ and some $\sring$ by $\Rec{\_}$. The latter notation is used to avoid confusion with the \emph{unweighted} recognizable tree languages ($\Rec{\mathbb{B}}$ in our notation).

\section{Definition of weighted alternating tree automata}
\label{sec:defs}

In this section, we define our weighted tree automata extended with the power of alternation. We give two definitions and show that they give rise to the same class of weighted tree languages. One definition follows the approach of~\cite{slutzki1985alternating-tre}, by partitioning the state set into two, in correspondence to existential and universal states. The other more closely aligns with the type of alternation originally introduced to finite automata by Chandra, Kozen, and Stockmeyer~\cite{chandra1981alternation}, in which transitions may refer to custom polynomials explicitly, thereby essentially allowing states to harness the power of both types of states simultaneously. We start by defining the latter version, as it will be our primary model.

\begin{definition}
\label{def:wata}
    A \emph{weighted alternating tree automaton (wata, or $(\sig, \sring)$-wata)} is a tuple $(\sig, Q, \sring, \pi, \delta)$, where $\sig$ is a ranked alphabet (of input symbols), $Q$ is a finite non-empty set (of states) with $Q \cap \sig = \emptyset$, $\sring$ is a semiring, $\pi \in \polynom{\sring}{\X_Q}$ is an \emph{initial polynomial}, and $\delta \colon Q \times \sig \to \polynom{\sring}{\X_{Q \times \nat}}$ is a \emph{transition mapping} such that, for each $k \geq 0$, $\sigma \in \sig^{(k)}$, and $q \in Q$, we have $\delta(q, \sigma) \in \polynom{\sring}{\X_{Q \times [k]}}$.
\end{definition}
The semantics of a wata $\aut = (\sig, Q, \sring, \pi, \delta)$ is the weighted tree language $\lang[\aut]\colon \trees_\sig \to \sring$ recognized by $\aut$, which we define as follows. First, for any state $q$ in $Q$ and tree $t = \sigma[t_1, \dots, t_k]$ in $\trees_\sig$, we let
\begin{equation} \label{eq:wata-sem}
    \lang[\aut;q](t) = \subst{\delta(q,\sigma)}{x_{(q',i)} \leftarrow \lang[\aut;q'](t_i) \mid q' \in Q, i \in [k]} \,.
\end{equation}
We then define, for any $t \in \trees_\sig$, $\lang[\aut](t) = \subst{\pi}{x_{q} \leftarrow \lang[{\aut;q}](t) \mid q \in Q}$.
We define the classes $\RecA{\sig, \sring}$, $\RecA{\sring}$, and $\RecA{\_}$ analogously to $\Rec{\sig, \sring}$, $\Rec{\sring}$, and $\Rec{\_}$.

The following example is a typical one highlighting the limitations of wta, and gives us an initial view of the power added by alternation.
\begin{example}
\label{ex:simple}
    Let $\sig_\mathrm{m} = \SET{\gamma^{(1)}, \alpha^{(0)}}$. Now consider the natural semiring $\nat$ and the weighted tree language $\lang\colon \trees_{\sig_\mathrm{m}} \to \nat$ defined for each $n \geq 0$ by $\lang(\gamma^n \alpha) = 2^{2^n}$. This language is not recognizable by any $(\sig, \sring)$-wta, since any function recognized by such a wta will be bounded from above by $b^n$ for some $b \in \nat$. However, letting $\aut_\mathrm{m} = (\sig_\mathrm{m}, \SET{q}, \nat, x_q, \delta)$ be a wata, in which $\delta(q,\gamma) = x_{(q,1)}^2$ and $\delta(q, \alpha) = 2$, we can compute $\lang[\aut_\mathrm{m}](\gamma^n \alpha)$ as a series of polynomial substitutions resulting in the expression $2^{2^n}$, i.e., $\lang[\aut_\mathrm{m}] = \lang$.
\end{example}

\begin{remark}
    There are at least two things about \examplename~\ref{ex:simple} worthy of note. The first being that over string alphabets, such as $\sig_\mathrm{m}$, wata reduce to the alternating weighted automata (over strings) of~\cite{kostolanyi2018alternating-wei}. The second being that the wata $\aut_\mathrm{m}$ is \emph{universal}. Although this property will not be examined in its own right in this paper, wata admit a rather natural notion of universality, namely when their transition mappings map only to monomials, leaving the addition operator of the semiring unused.
\end{remark}

\subsection{Normal forms}

We now introduce a couple of normal forms of wata, that will provide helpful assumptions in proofs. We say that a wata with state set $Q$ is \emph{initial-polynomial normalized} if its initial polynomial is of the form $\sum_{q \in Q'} x_{q}$ for some $Q' \subseteq Q$.

\begin{lemma}
\label{lem:normalized}
    For every $(\sig, \sring)$-wata $\aut$, there is an initial-polynomial normalized $(\sig, \sring)$-wata $\aut'$ with $\lang[\aut'] = \lang[\aut]$.
\end{lemma}
\begin{proof}
    Let $\aut = (\sig, Q, \sring, \pi, \delta)$ be a wata and let $m_1, \dots, m_\ell$ be distinct monomials, the sum of which make up the initial polynomial $\pi$. The main idea behind the construction is to postpone the computation of the $m_i$ to the following computation step. Formally, we proceed with the construction of $\aut' = (\sig, Q', \sring, \pi', \delta')$ as follows. First, let $P = \SETC{p_{(m_i, \sigma)}}{i \in [\ell], \sigma \in \sig}$ (which we assume is disjoint from $Q$) and set $Q' = Q \cup P$. Next, we set $\pi' = \sum_{p \in P} x_p$ and define $\delta'$ by
    \begin{equation*}
        \delta'(q,\sigma) =
        \begin{cases}
            \delta(q, \sigma) & \text{if } q \in Q \\
            \subst{m}{x_{q'} \leftarrow \delta(q', \sigma) \mid q' \in Q} & \text{if } q = p_{(m, \sigma)} \\
            0 & \text{otherwise.}
        \end{cases}
    \end{equation*}
    Obviously, $\aut'$ is initial-polynomial normalized. The proof of $\lang[\aut'] = \lang[\aut]$ is straightforward.
\end{proof}

The next property concerns the degree of the polynomials appearing in a wata.

\begin{lemma}
\label{lem:degree1}
    For every $(\sig, \sring)$-wata $\aut$, there is a wata $\aut' = (\sig, Q', \sring, \pi, \delta')$ with $\lang[\aut'] = \lang[\aut]$ and $\delta'(q,\sigma)$ having degree at most $1$ for all $q \in Q'$ and $\sigma \in \sig$.
\end{lemma}
\begin{proof}
    Let $\aut = (\sig, Q, \sring, \pi, \delta)$ be a wata and let $D$ be the maximal degree of any polynomial in the image of $\delta$. The idea is simply to make $D$ copies of each state in $Q$, and to modify $\delta$ to utilize these copies instead of exponents. We define $\aut' = (\sig, Q', \sring, \pi, \delta')$ by first letting $Q' = \SETC{\langle q, d \rangle}{q \in Q, d \in [D]}$; then, for every $k \geq 0$, $\sigma \in \sig^{(k)}$, $q \in Q$, and $d \in [D]$, we let
    \begin{equation*}
        \delta'(\langle q, d \rangle, \sigma) = \sum_{i \in [\ell]} c_i \cdot \prod_{(q',j) \in Q \times [k]} x_{( \langle q', 1 \rangle, j)} \cdots x_{( \langle q', d_{(i,q',j)} \rangle, j)} \,,
    \end{equation*}
    where $\delta(q, \sigma) = \sum\limits_{i \in [\ell]} c_i \cdot \prod\limits_{(q',j) \in Q \times [k]} x_{(q', j)}^{d_{(i,q',j)}}$.
    It should be clear that $\lang[\aut'] = \lang[\aut]$.
\end{proof}

Note that the construction in \lemmaname~\ref{lem:degree1} does not affect the initial polynomial, i.e., the subject of \lemmaname~\ref{lem:normalized}. Hence we can assume that any wata can be transformed into an equivalent wata possessing both of these properties; such wata will be referred to as \emph{normalized}.

\subsection{Alternative definition}
\label{sec:two-mode}

\definitionname~\ref{def:wata} deviates from the approach to alternation taken in~\cite{slutzki1985alternating-tre} by not admitting a partitioning of the state set into existential and universal states. The benefit of this is a formalism that is often more elegant and convenient to work with, although the ties to this other variety of alternation may not be directly apparent. To this end, we introduce two-mode wata,\footnote{Adhering to the nomenclature of \cite{kostolanyi2018alternating-wei}.} a model expressively equivalent to wata, establishing this connection. In the weighted setting, existentiality is extended to summation, and universality to product.

\begin{definition}
    A \emph{two-mode weighted alternating tree automaton (two-mode wata, or two-mode $(\sig, \sring)$-wata)} is a tuple $(\sig, Q_\oplus, Q_\otimes, R, \sring, \iota, \wt)$, where $\sig$ is a ranked alphabet, $Q_\oplus$ and $Q_\otimes$ are finite sets (of sum and product states, respectively) with $Q_\oplus \cap Q_\otimes = \sig \cap (Q_\oplus \cup Q_\otimes) = \emptyset$, $R$ is a ranked alphabet (of transition rules) such that, for each $k \geq 0$, $R^{(k)} \subseteq (Q_\oplus \cup Q_\otimes) \times \sig^{(k)} \times (Q_\oplus \cup Q_\otimes)^k$, $\sring$ is a semiring, $\iota\colon (Q_\oplus \cup Q_\otimes) \to \sring$ is an initial weight mapping, and $\wt\colon R \to \sring$ is a rule weight mapping.
\end{definition}
A two-mode wata $\aut = (\sig, Q_\oplus, Q_\otimes, R, \sring, \iota, \wt)$ recognizes the weighted tree language $\lang[\aut]\colon \trees_\sig \to \sring$ defined as follows. For any $q\in Q_\oplus \cup Q_\otimes$ and $\sigma \in \sig$, let $R_{q,\sigma}$ denote the set of all rules in $R$ of the form $(q,\sigma, w)$. Now, for any state $q \in Q_\oplus \cup Q_\otimes$ and tree $t = \sigma[t_1, \dots, t_k]$ in $\trees_\sig$, we let
\begin{equation} \label{eq:2mode-sem}
    \lang[\aut;q](t) =
    \begin{dcases}
        \sum\limits_{\rho = (q,\sigma,q_1\dots q_k) \in R} \wt(\rho) \cdot \prod_{i \in [k]} \lang[\aut;q_i](t_i) &\text{if }  q \in Q_\oplus \\
        \prod\limits_{\rho = (q,\sigma,q_1\dots q_k) \in R} \wt(\rho) \cdot \prod_{i \in [k]} \lang[\aut;q_i](t_i) &\text{if } q \in Q_\otimes \text{ and } R_{q,\sigma} \neq \emptyset \\
        0 &\text{otherwise.}
    \end{dcases}
\end{equation}
We then define, for any $t \in \trees_\sig$, $\lang[\aut](t) = \sum_{q \in (Q_\oplus \cup Q_\otimes)} \iota(q) \cdot \lang[{\aut;q}](t)$.

The difference going from wata to two-mode wata can be viewed as externalizing some of the semantics; whereas an entire computation is essentially built into the transitions of wata, only the \emph{mode} of the computation is indicated in the transitions of two-mode wata. The next result shows that this difference does not alter the expressive power of the automata.

\begin{theorem}
\label{thm:diff-sem-same-class}
    For any two-mode $(\sig, \sring)$-wata $\aut$, there is a $(\sig, \sring)$-wata $\aut'$ with $\lang[\aut'] = \lang[\aut]$, and vice versa.
\end{theorem}
\begin{proof}
    Given a two-mode wata $\aut = (\sig, Q_\oplus, Q_\otimes, R, \sring, \iota, \wt)$ one can easily construct an ordinary wata $\aut' = (\sig, Q, \sring, \pi, \delta)$ with $\lang[\aut'] = \lang[\aut]$ by setting

    \begin{itemize}
        \item $Q = Q_\oplus \cup Q_\otimes$
        \item $\pi = \sum_{q \in Q} \iota(q) \cdot x_q$, and
    \end{itemize}
    \begin{equation*}
        \delta(q, \sigma) =
        \begin{dcases}
            \sum\limits_{\rho = (q, \sigma, q_{1}\dots q_{k}) \in R} \wt(\rho) \cdot \prod\limits_{i \in [k]} x_{(q_i, i)} & \text{if } q \in Q_\oplus \\
            \prod\limits_{\rho = (q, \sigma, q_{1}\dots q_{k}) \in R} \wt(\rho) \cdot \prod\limits_{i \in [k]} x_{(q_i, i)} & \text{if } q \in Q_\otimes \text{ and } R_{q,\sigma} \neq \emptyset \\
            0 & \text{otherwise.}
        \end{dcases}
    \end{equation*}
    In light of Equations \eqref{eq:wata-sem} and \eqref{eq:2mode-sem}, it should be clear that $\lang[\aut'] = \lang[\aut]$.

    The other direction is a bit more involved. We employ a guess-and-verify strategy in which the levels of a tree are forced to essentially alternate between sum and product states; in sum states, the automaton pre-computes the cost of processing all possible direct descendants, whereas it in product states filters correct computations from incorrect ones and distributes the subsequent computations.
    
    Given a wata $\aut' = (\sig, Q, \sring, \pi, \delta)$ --- assumed to be normalized, by \lemmaname s~\ref{lem:normalized} and \ref{lem:degree1} --- we construct the two-mode wata $\aut = (\sig, Q_\oplus, Q_\otimes, R, \sring, \iota, \wt)$ as follows.

    First, let $M$ be the (finite) set of all monomials in $\polynom{\sring}{\X_{Q \times [\mxrk(\sig)]}}$ with coefficient $1$ and of degree at most $ D = |Q| \cdot \mxrk(\sig)$ --- this is the highest possible degree of any polynomial resulting from a substitution over polynomials in $\Ima(\delta)$. Note that the only constant monomial in $M$ is the monomial $1$. Let $Q_\oplus = Q \times [D]$ and $Q_\otimes = \SET{p_r} \cup \SETC{p_{(m, \sigma)}}{m \in M, \sigma \in \sig}$, where $p_r$ is distinct from all other states in $Q_\oplus \cup Q_\otimes$. Assuming $\pi = \sum_{q \in Q'} x_q$ for some $Q' \subseteq Q$, we then define $\iota(\langle q, 1 \rangle) = 1$ if $q \in Q'$, and $\iota(q') = 0$ for all other $q' \in Q_\oplus \cup Q_\otimes$.
    
    Lastly, we define $R$ and $\wt$; as for wta, we write $l \xrightarrow{c} r$ to denote rules together with their weight.
    Starting with rules originating from sum states, we introduce --- for each $k \geq 0$, $\sigma \in \sig^{(k)}$, $q \in Q$, and $\sigma_1, \dots, \sigma_k \in \sig$ --- the notation $\varphi_{q, \sigma, \langle \sigma_1, \dots, \sigma_k \rangle} = \subst{\delta(q,\sigma)}{x_{(q',i)} \leftarrow \delta(q', \sigma_i) \mid q' \in Q, i \in [k]}$.
    Intuitively, $\varphi_{q, \sigma, \langle \sigma_1, \dots, \sigma_k \rangle}$ represents the cost of, from $q$, processing both the current and the next level of the tree, \emph{provided} the children nodes are indeed labeled by $\sigma_1, \dots, \sigma_k$. We may assume that each non-zero $\varphi_{q, \sigma, \langle \sigma_1, \dots, \sigma_k \rangle}$ is of the form $c_1 \cdot m_1 + \dots + c_\ell \cdot m_\ell$, where $c_1, \dots, c_\ell \in \sring \setminus \SET{0}$, $m_1, \dots, m_\ell \in M$, and $m_i \neq m_j$ whenever $i \neq j$; for constant polynomials, we have $\ell = 1$ and $m_1 = 1$. For each non-zero $\varphi_{q, \sigma, \langle \sigma_1, \dots, \sigma_k \rangle} = c_1 \cdot m_1 + \dots + c_\ell \cdot m_\ell$, each $i \in [\ell]$, and each $d \in [D]$, we let $R$ contain the rule
    \begin{equation*}
        \langle q, d \rangle [\sigma] \xrightarrow{c_i} \langle p_{(m_i, \sigma_1)}, \dots, p_{(m_i, \sigma_k)} \rangle \,.
    \end{equation*}

    Moving on to rules originating from product states, we note that the inclusion of the state $p_r$ in $Q_\otimes$ is a technicality, and its purpose is simply to let trees pass through without affecting their weights; to this end, we let $R$ contain the rule
        $p_r[\sigma] \xrightarrow{1} \langle \underbrace{p_r, \dots, p_r}_{k} \rangle$
    for each $k \geq 0$ and $\sigma \in \sig^{(k)}$. The other product states are then the ones in charge of verifying the computations from the parent node and distributing the computation appropriately: for any $k \geq 0$, ${\sigma \in \sig^{(k)}}$, and $m \in M$, we define $R_{p_{(m, \sigma)}; \sigma}$ as follows. In the case where $k = 0$ or $m = 1$, we let $R_{p_{(m, \sigma)}; \sigma}$ consist of the single rule $p_{(m, \sigma)}[\sigma] \xrightarrow{1} \langle \underbrace{p_r, \dots, p_r}_{k} \rangle$. Otherwise, we assume $m = \prod_{(q', i) \in Q \times [k]}\limits x_{(q',i)}^{d_{(q',i)}}$ and let $R_{p_{(m, \sigma)}; \sigma}$ contain, for each $q \in Q$, $i \in [k]$, and $d \in [d_{(q,i)}]$, the rule $p_{(m, \sigma)}[\sigma] \xrightarrow{1} \langle \underbrace{p_r, \dots, p_r}_{i-1}, \langle q, d\rangle, \underbrace{p_r, \dots, p_r}_{k-i} \rangle$.
    This concludes the definition of $R$ and $\wt$, and thereby $\aut$. Note in particular that $R_{p_{(m,\sigma)}; \sigma'} = \emptyset$ whenever $\sigma \neq \sigma'$.

    To show $\lang[\aut] = \lang[\aut']$, one can prove the following two claims via simultaneous structural induction over an input tree $t = \sigma[t_1, \dots, t_k] \in \trees_\sig$:
    \begin{enumerate}[label=(\roman*)]
        \item For all $m \in M$ and $\sigma' \in \sig$,
        \begin{equation*}
            \lang[\aut; p_{(m, \sigma')}](t) =
            \begin{cases}
                \subst{m}{x_{(q, i)} \leftarrow \lang[\aut'; q] (t_i) \mid q \in Q, i \in [k]} &\text{if } \sigma' = \sigma \text{ and } m \in \polynom{\sring}{\X_{Q \times [k]}} \\
                0 &\text{otherwise.}
            \end{cases}
        \end{equation*}
        \item\label{it:2mode-claim2} For all $q \in Q$ and $d \in [D]$, $\lang[\aut; \langle q,d \rangle] (t) = \lang[\aut'; q] (t)$.
    \end{enumerate}
    The result then follows from \claimname~\ref{it:2mode-claim2} and the definition of $\iota$. We omit the explicit inductive proof due to space constraints.
    \end{proof}

\section{Expressive power}
\label{sec:power}

It is easy to see that ordinary weighted tree automata are special cases of wata; a wta is simply a (normalized) wata for which all transition polynomials $\delta(q, \sigma)$, for $\sigma \in \sig^{(k)}$, are sums of monomials of the form $c \cdot x_{(q_1, 1)} \cdots x_{(q_k, k)}$. We know already from \examplename~\ref{ex:simple} that this restriction is not just syntactical; there are semirings over which wata are strictly more powerful than wta. On the other hand, we know from the unweighted case~\cite{slutzki1985alternating-tre} that there are semirings (i.e., the boolean semiring) over which wata and wta are equivalent in terms of expressive power. The latter semirings are below shown to be exactly the locally finite ones.

\newcommand{\subsring}{\mathbb{K}_\mathrm{sub}}

\begin{theorem}
\label{thm:eq-iff-loc-fin}
    $\Rec{\sring} = \RecA{\sring}$ if and only if $\sring$ is locally finite.
\end{theorem}
\begin{proof}
    The ``only if'' direction is a direct consequence of~\cite[Thm~$7.1$]{kostolanyi2018alternating-wei}, in which it is shown that $\RecA{\sig, \sring} = \Rec{\sig, \sring}$ if and only if $\sring$ is locally finite, \emph{for string alphabets} $\sig$. Thus, in particular, $\Rec{\sring} \subsetneq \RecA{\sring}$ whenever $\sring$ is \emph{not} locally finite. What remains to be shown is that $\RecA{\sring} \subseteq \Rec{\sring}$ whenever $\sring$ is locally finite.

    Let $\sring$ be locally finite and $\aut = (\sig, Q, \sring, \pi, \delta)$ be a wata, which we by \lemmaname~\ref{lem:normalized} assume to be initial-polynomial normalized. First off, let $\sringset'$ be the (finite) set of all coefficients appearing in polynomials in $\Ima(\delta)$, and let $\subsring$ be the semiring $\langle \sringset' \rangle$. Since $\sring$ is locally finite, $\subsring$ is finite, and $\aut$ can be viewed as a $(\sig, \subsring)$-wata. Since $\subsring$ is finite, there are finitely many functions $\subsring^{n} \to \subsring$ for any given $n \in \nat$; in particular, there are finitely many functions defined by polynomials in $\polynom{\subsring}{\X_Q}$. For any two polynomials $p_1$ and $p_2$ in $\polynom{\subsring}{\X_Q}$, we let $p_1 \sim p_2$ if and only if $p_1$ and $p_2$ define the same polynomial function, and let $M$ denote the set of all equivalence classes $[m]_\sim$, where $m$ is a monomial in $\polynom{\subsring}{\X_Q}$ with coefficient $1$. In what follows, we assume that each equivalence class in $M$ has some canonical monomial $m$ and consequently denote the class by $[m]_\sim$.
    
    We will now construct a wta $\aut' = (\sig, M, M_d, R, \subsring, \wt)$ such that $\lang[\aut'] = \lang[\aut]$. First, assuming $\pi = x_{q_1} + \cdots + x_{q_n}$, we let $M_d = \SET{[x_{q_1}]_\sim, \dots, [x_{q_n}]_\sim}$. Next, we define $R$, and we do this by, for each $k \geq 0$, $\sigma \in \sig^{(k)}$, and $[m]_\sim \in M$, defining each $R_{[m]_\sim ; \sigma}$ individually:
    
    The polynomial $\subst{m}{x_q \leftarrow \delta(q, \sigma) \mid q \in Q}$ can be assumed to be of the form $\sum_{j \in [\ell]} c_j \cdot m_{j_1} \cdots m_{j_k}$, where $c_j$ is in $\subsring$, each $m_{j_i}$ is a monomial in $\polynom{\subsring}{\X_{Q \times \SET{i}}}$ with coefficient $1$, and $m_{j_1} \cdots m_{j_k} \neq m_{j'_1} \cdots m_{j'_k}$ for $j \neq j'$. For each $j \in [\ell]$ and $i \in [k]$, we let $m'_{j_i} = \subst{m_{j_i}}{x_{(q,i)} \leftarrow x_q \mid q \in Q}$. For each $j \in [\ell]$, we then let $R_{[m]_\sim ; \sigma}$ contain the rule $\rho_j = ([m]_\sim, \sigma, \langle [m'_{j_1}]_\sim, \dots, [m'_{j_k}]_\sim \rangle)$ and set $\wt(\rho_j) = c_j$.

    One can now prove, for each $t \in \trees_\sig$ and $m \in M$, that $\lang[{\aut';[m]_\sim}] (t) = \subst{m}{x_q \leftarrow \lang[\aut;q] (t) \mid q \in Q}$ via structural induction on $t$. The result will then follow, since
    \begin{equation*}
        \lang[\aut'](t) = \sum_{[x_q]_\sim \in M_d} \lang[{\aut'; [x_q]_\sim }] (t) = \sum_{[x_q]_\sim \in M_d} \lang[\aut; q] (t) =  \subst{\pi}{x_q \leftarrow \lang[\aut; q] (t) \mid q \in Q} = \lang[\aut](t) \,.
    \end{equation*}
    We omit this part of the proof because of a lack of space.
\end{proof}

Another related automaton model appearing in the literature is the \emph{polynomially weighted tree automata (pol-wta)}~\cite{seidl1994finite-tree-aut,borchardt2005bounds-for-tree}. These can be described as initial-polynomial normalized wata for which all transition polynomials $\delta(q,\sigma)$, for $\sigma \in \sig^{(k)}$, are sums of monomials of the form $c \cdot x_{(q_1, 1)}^{d_1} \cdots x_{(q_k, k)}^{d_k}$. Obviously, any wta is a pol-wta, so letting $\RecPol{\sring}$ denote the class of weighted tree languages recognized by pol-wta over $\sring$, we have $\Rec{\sring} \subseteq \RecPol{\sring} \subseteq \RecA{\sring}$. \theoremname~\ref{thm:eq-iff-loc-fin} tells us that $\Rec{\sring} = \RecPol{\sring} = \RecA{\sring}$ for any locally finite $\sring$. However, the wata in \examplename~\ref{ex:simple} is a pol-wta, showing that not every pol-wta can be converted into an equivalent wta, i.e., $\Rec{\sring} \subsetneq \RecPol{\sring}$ for some (not locally finite) $\sring$. This leaves the question of whether $\RecPol{\sring} = \RecA{\sring}$ for all $\sring$. The answer is negative, as shown in the example below.

\begin{example}
\label{ex:wata-vs-pol-wta}
    Recall $\sig_\mathrm{m} = \SET{\gamma^{(1)}, \alpha^{(0)}}$. Let $\aut = (\sig_\mathrm{m}, \SET{q_{!}, q_{+}}, \nat, x_{q_!}, \delta)$ be a wata and
    \begin{equation*}
    \begin{aligned}
        & \delta(q_!, \gamma) = x_{(q_!, 1)} + x_{(q_!, 1)} \cdot x_{(q_+, 1)} &&\qquad \delta(q_+, \gamma) = 1 + x_{(q_+, 1)} \\
        & \delta(q_!, \alpha) = 1 &&\qquad \delta(q_+, \alpha) = 0 \,.
    \end{aligned}
    \end{equation*}
    For every $n \geq 0$, we have $\lang[\aut](\gamma^n \alpha) = n!$.
    
    Assume, toward a contradiction, that there is a pol-wta $\aut' = (\sig_\mathrm{m}, Q, \nat, \pi, \delta')$ such that $\lang[\aut'] = \lang[\aut]$. For any $b \in \nat$, we have that $n! \notin \bigO(b^n)$, and in order to avoid $\lang[\aut'] \in \bigO(b^n)$, we can via a pumping-like argument conclude that there need to be integers $m, \ell, \ell'$ such that $1 \leq \ell \leq \ell' \leq m$ and states $q_0, \dots, q_m = q_\ell$ in $Q$ such that: (i) $x_{q_0}$ appears in $\pi$; (ii) $x_{(q_{i}, 1)}$ appears in $\delta(q_{i-1}, \gamma)$ for each $i \in [m]$; (iii) $x_{(q_{\ell'}, 1)}$ appears in $\delta(q_{\ell' - 1}, \gamma)$ with exponent $d > 1$; and (iv) $\lang[\aut'; q_{\ell}] (\gamma^n \alpha) > 1$ for any $n \geq m - \ell$. These conditions taken together, however, mean that $\lang[\aut'] \in \bigOmega(2^{2^n})$, contradicting the fact that $n! \notin \bigOmega(2^{2^n})$.
\end{example}

\examplename~\ref{ex:wata-vs-pol-wta} yields the result below.

\begin{proposition}
\label{prop:pol-wta-weaker}
    There are semirings $\sring$ such that $\RecPol{\sring} \subsetneq \RecA{\sring}$, i.e., over which polynomially weighted tree automata are strictly weaker than weighted alternating tree automata.
\end{proposition}

\section{Closure properties}
\label{sec:closure}

To get a better sense of the power and limitations of wata, we look at some operations on weighted tree languages constituting potential closure properties for the weighted tree languages defined by wata.

\subsection{Basic operations}

A class $C$ of weighted tree languages is said to be \emph{closed under}
\begin{itemize}
    \item \emph{sum} if, for every $\lang_1$ and $\lang_2$ in $C$ with $\lang_i \colon \trees_\sig \to \sring$, we have that $\lang_1 + \lang_2 \in C$, where $(\lang_1 + \lang_2)(t) = \lang_1(t) + \lang_2(t)$ for each $t \in \trees_\sig$,
    \item \emph{Hadamard product} if, for every $\lang_1$ and $\lang_2$ in $C$ with $\lang_i \colon \trees_\sig \to \sring$, we have that $\lang_1 \cdot \lang_2 \in C$, where $(\lang_1 \cdot \lang_2)(t) = \lang_1(t) \cdot \lang_2(t)$ for each $t \in \trees_\sig$,
    \item \emph{scalar multiplication} if for every $\lang \colon \trees_\sig \to \sring$ in $C$ and $a \in \sring$, we have that $a \cdot \lang \in C$, where $(a \cdot \lang)(t) = a \cdot \lang(t)$ for each $t \in \trees_\sig$, and
    \item \emph{top-concatenation} if, for every $k \in \nat$, $\sigma \in \sig^{(k)}$, and $\lang_1, \dots, \lang_k$ in $C$ with $\lang_i \colon \trees_\sig \to \sring$, we have that $\mathrm{top}_\sigma (\lang_1, \dots, \lang_k) \in C$, where
    \begin{equation*}
        \mathrm{top}_\sigma (\lang_1, \dots, \lang_k) (t) =
        \begin{dcases}
            \prod_{i \in [k]} \lang_i(t_i) &\text{if } t = \sigma[t_1, \dots, t_k]\\
            0 &\text{otherwise,}
        \end{dcases}
    \end{equation*}
    for each $t \in \trees_\sig$.
\end{itemize}
The above closure properties, while interesting in themselves, will also help showcase the flexibility of the wata model; in the proof below, constructions are more or less trivial.

\begin{proposition}
    $\RecA{\sring}$ is closed under:
    \begin{enumerate}[label=(\alph*)]
        \item\label{prop:closure-sum} Sum (cf.\ \cite[Lem.~$1$]{drewes2025dynamically-wei})
        \item\label{prop:closure-hadamard} Hadamard product (cf.\ \cite[Lem.~$1$]{drewes2025dynamically-wei})
        \item\label{prop:closure-scalar} Scalar multiplication
        \item\label{prop:closure-topconcat} Top-concatenation
    \end{enumerate}
\end{proposition}
\begin{proof}
     For any $m \in \nat$ and $i \in [m]$, let $\aut_i = (\sig, Q_i, \sring, \pi_i, \delta_i)$ be arbitrary wata with $\bigcap_{i \in [m]} Q_i = \emptyset$. We give here, for each property \ref{prop:closure-sum}--\ref{prop:closure-topconcat}, the construction of a wata $\aut = (\sig, Q, \sring, \pi, \delta)$ demonstrating the respective closure property, but leave the correctness verification to the reader.
    \begin{itemize}[align=left]
        \item[\ref{prop:closure-sum}--\ref{prop:closure-hadamard}] $Q = Q_1 \cup Q_2$, $\pi = \pi_1 + \pi_2$ for \ref{prop:closure-sum} and $\pi = \pi_1 \cdot \pi_2$ for \ref{prop:closure-hadamard}, and, for each $q \in Q$ and $\sigma \in \sig$,
        \begin{equation*}
            \delta(q, \sigma) =
            \begin{cases}
                \delta_1(q, \sigma) & \text{if } q \in Q_1 \\
                \delta_2(q, \sigma) & \text{if } q \in Q_2 \,.
            \end{cases}
        \end{equation*}
        
        \item[\ref{prop:closure-scalar}] Given $a \in \sring$, we let $Q = Q_1$, $\pi = a \cdot \pi_1$, and $\delta = \delta_1$.

        \item[\ref{prop:closure-topconcat}] For any $k \in \nat$ and $\sigma \in \sig^{(k)}$, we let $Q = \bigcup_{i \in [k]} Q_i \cup \SET{q_0}$, where $q_0$ is fresh, $\pi = x_{q_0}$, and, for each $q \in Q$ and $\gamma \in \sig$,
        \begin{equation*}
        \begin{gathered}[b]
            \delta(q, \gamma) =
            \begin{dcases}
                \prod_{i \in [k]} \subst{\pi_i}{x_{q'} \leftarrow x_{(q',i)} \mid q' \in Q_i} & \text{if } q = q_0 \text{ and } \gamma = \sigma \\
                \delta_i (q, \gamma) & \text{if } q \in Q_i \quad (i \in [k]) \\
                0 & \text{otherwise.}
            \end{dcases}
        \\
        \end{gathered}
        \qedhere
        \end{equation*}
    \end{itemize}
\end{proof}

\subsection{Tree homomorphisms}

Given two ranked alphabets $\sig$ and $\Delta$, a \emph{tree homomorphism} (or \emph{$(\sig, \Delta)$-tree homomorphism}) is a family $h = (h_k \mid k \in \nat)$ of functions $h_k\colon \sig^{(k)} \to \trees_\Delta(\Z_{[k]})$. The family $h$ is extended into a function ${h\colon \trees_\sig \to \trees_\Delta}$ recursively via $h(\sigma[t_1, \dots, t_k]) = \subst{h_k(\sigma)}{z_i \leftarrow h(t_i) \mid i \in [k]}$ for all $k \geq 0$ and $\sigma \in \sig^{(k)}$. A $(\sig, \Delta)$-tree homomorphism $h$ is called \emph{linear} (respectively, \emph{non-deleting}) if $h_k(\sigma)$ is linear (respectively, non-deleting) with respect to $\Z_{[k]}$ for all $k \geq 0$ and $\sigma \in \sig^{(k)}$, and \emph{productive} if $h_k(\sigma) \notin \Z_{[k]}$ for all $k \geq 0$ and $\sigma \in \sig^{(k)}$. Furthermore, $h$ is called \emph{deterministically relabeling}\footnote{Commonly called a \emph{deterministic tree relabeling}.} if, for each $k \geq 0$ and $\sigma \in \sig^{(k)}$, the tree $h_k(\sigma)$ is of the form $\gamma[z_1, \dots, z_k]$ for some $\gamma \in \Delta^{(k)}$. Note that any deterministically relabeling tree homomorphism is linear, non-deleting, and productive. We denote the class of all tree homomorphisms by $\HOM$.

Let $\lang\colon \trees_\sig \to \sring$ be a weighted tree language and $h$ a $(\sig, \Delta)$-tree homomorphism such that $\sring$ is \emph{(countably) complete}\footnote{Roughly, infinite sums are defined and they are associative, commutative, and satisfy the distributivity laws. For the precise definition see, e.g., \cite{droste2009semirings-and-f}.} or $h^{-1}(s)$ is finite for all $s \in \trees_\Delta$ (the latter condition always holds if $h$ is non-deleting and productive). The \emph{application of $h$ to $\lang$} is the weighted tree language $\chi(h)(\lang)\colon \trees_\Delta \to \sring$ defined, for each $s \in \trees_\Delta$, by
\begin{equation*}
    \chi(h)(\lang)(s) = \sum_{t \in h^{-1}(s)} \lang(t)\, .
\end{equation*}
Now, we say that a class $C$ of weighted tree languages is \emph{closed under tree homomorphisms in $H \subseteq \HOM$} if $\chi(h)(\lang)$ is in $C$ for all $\lang$ in $C$ and all tree homomorphisms $h$ in $H$.
It is known that $\Rec{\sring}$ is closed under linear, non-deleting, and productive tree homomorphisms~\cite[Cor.~$12.11.2$]{fulop2022weighted-tree-a}, as well as linear and non-deleting tree homomorphisms given that $\sring$ is complete~\cite[Thm.~$5.3$]{fulop2011weighted-extend}. Due to \theoremname~\ref{thm:eq-iff-loc-fin} we know that the same is true for $\RecA{\sring}$ if $\sring$ is locally finite. However, for semirings that are not locally finite, the situation is different, as illustrated by the following example, adapted from~\cite{grabolle2023a-nivat-theorem}.

\begin{example}
\label{ex:nonclosure-homom}
    Consider the polynomial semiring $\polynom{\mathbb{B}}{\X_{[1]}}$, and let $\sig_1 = \SET{a^{(1)}, c^{(1)}, d^{(1)}, \#^{(1)}, \bot^{(0)}}$ and $\sig_2 = \SET{a^{(1)}, b^{(1)}, \#^{(1)}, \bot^{(0)}}$. Next, let $\lang_1 \colon \trees_{\sig_1} \to \polynom{\mathbb{B}}{\X_{[1]}}$ and $\lang_2 \colon \trees_{\sig_2} \to \polynom{\mathbb{B}}{\X_{[1]}}$ be defined by
    \begin{equation*}
        \lang_1(t) =
            \begin{cases}
                x_1^{ki} &\text{if } t = a^i \# c^k d^\ell \bot \quad (i,k,\ell \geq 0) \\
                0 &\text{otherwise,}
            \end{cases}
        \qquad
        \lang_2(t) =
            \begin{cases}
                \sum_{k = 0}^j x_1^{ki} &\text{if } t = a^i \# b^j \bot \quad (i,j \geq 0) \\
                0 &\text{otherwise,}
            \end{cases}
    \end{equation*}
    and let $h\colon \trees_{\sig_1} \to \trees_{\sig_2}$ be a tree homomorphism defined by $h(a) = a[z_1]$, $h(c) = h(d) = b[z_1]$, $h(\#) = \#[z_1]$, and $h(\bot) = \bot$. Note that $h$ is deterministically relabeling.
    While $\lang_1 \in \RecA{\polynom{\mathbb{B}}{\X_{[1]}}}$, we have that $\chi(h)(\lang_1) = \lang_2 \notin \RecA{\polynom{\mathbb{B}}{\X_{[1]}}}$~\cite[Prop.~$5.3$]{grabolle2023a-nivat-theorem}.
\end{example}

    \examplename~\ref{ex:nonclosure-homom} shows that, over some semirings, the class of weighted tree languages recognized by wata is not closed under tree homomorphisms, even if these are heavily constrained.

\begin{proposition}
    There are semirings $\sring$ for which $\RecA{\sring}$ is not closed under deterministically relabeling tree homomorphisms.
\end{proposition}

\subsection{Inverse tree homomorphisms}

A class $C$ of weighted tree languages is \emph{closed under inverse tree homomorphisms in $H \subseteq \HOM$} if $h \mathrel{;} \lang$ is in $C$ for all tree homomorphisms $h$ in $H$ and all $\lang$ in $C$, i.e., $H \mathrel{;} C \subseteq C$.

$\Rec{\sring}$ is known to be closed under inverse \emph{linear} tree homomorphisms~\cite[Thm.~$12.12.2$]{fulop2022weighted-tree-a}, but is not in general closed under inverse tree homomorphisms~\cite[Thm.~$7.1$]{grabolle2023a-nivat-theorem}. We show here that this limitation does not apply to $\RecA{\sring}$.

\begingroup
\microtypesetup{protrusion=true,expansion=true}%
\begin{theorem}
\label{thm:closed-inv-homom}
    $\RecA{\sring}$ is closed under inverse tree homomorphisms, i.e., $\HOM \mathrel{;} \RecA{\sring} \subseteq \RecA{\sring}$.
\end{theorem}
\endgroup
\begin{proof}
    Let $\aut = (\Delta, Q, \sring, \pi, \delta)$ be a wata and $h\colon \trees_\sig \to \trees_\Delta$ a tree homomorphism. We construct a wata $\aut' = (\sig, Q, \sring, \pi, \delta')$ such that $\lang[\aut'] = h \mathrel{;} \lang[\aut]$ as follows.

    We first extend $\aut$ to a wata $\aut_\Z$ capable of handling the necessary variables. For this, we let $m = \mxrk(\sig)$ and $\Delta_\Z = \Delta \cup \Z_{[m]}$, in which the variables in $\Z_{[m]}$ are viewed as symbols of rank $0$. We then define the wata $\aut_\Z = (\Delta_\Z, Q, \polynom{\sring}{\X_{Q \times [m]}}, \pi, \delta_\Z)$ over the ranked alphabet $\Delta_\Z$ and the semiring $\polynom{\sring}{\X_{Q \times [m]}}$ by letting $\delta_\Z(q,z_i) = x_{(q,i)}$ and $\delta_\Z(q,\gamma) = \delta(q,\gamma)$ for all $q \in Q$, $i \in [m]$, and $\gamma \in \Delta$. Obviously, $\lang[\aut ; q](t) = \lang[\aut_\Z ; q](t)$ for any $q \in Q$ and $t \in \trees_\Delta$.
    Next, to define the transition function $\delta'$ of $\aut'$, we simply let $\delta'(q,\sigma) = \lang[\aut_\Z; q](h_k(\sigma))$ for all $q \in Q$, $k \geq 0$, and $\sigma \in \sig^{(k)}$.

    Since $\aut$ and $\aut'$ have the same initial polynomial, it is sufficient to show $\lang[\aut'; q](t) = \lang[\aut; q](h(t))$ for all $q \in Q$ and $t \in \trees_\sig$, which we do via structural induction on $t = \sigma[t_1, \dots, t_k]$:

    For $k = 0$, we have that
    \begin{equation*}
        \lang[\aut'; q](\sigma) = \delta'(q, \sigma) = \lang[\aut_\Z; q](h_0(\sigma)) = \lang[\aut_\Z; q](h(\sigma)) = \lang[\aut; q](h(\sigma)),
    \end{equation*}
    since $h(\sigma) \in \trees_\Delta$. Now, assuming $k > 0$ and $\lang[\aut'; q'](t_i) = \lang[\aut; q'](h(t_i))$ for all $q' \in Q$ and $i \in [k]$, we get
    \begin{equation*}
    \begin{aligned}
        \lang[\aut'; q](\sigma[t_1, \dots, t_k])
        &= \subst{\delta'(q, \sigma)}{x_{(q',i)} \leftarrow \lang[\aut'; q'](t_i) \mid q' \in Q, i \in [k]} && \\
        &= \subst{\delta'(q, \sigma)}{x_{(q',i)} \leftarrow \lang[\aut; q'](h(t_i)) \mid q' \in Q, i \in [k]} &&\text{(by assumption)} \\
        &= \subst{\lang[\aut_\Z ; q](h_k(\sigma))}{x_{(q',i)} \leftarrow \lang[\aut; q'](h(t_i)) \mid q' \in Q, i \in [k]} && \\
        &= \lang[\aut; q] \left( \subst{h_k(\sigma)}{z_{i} \leftarrow h(t_i) \mid i \in [k]} \right) &&\text{(by \claimname~\ref{claim:closed-inv-homom-helper})} \\
        &= \lang[\aut; q](h(\sigma[t_1, \dots, t_k])) \,. &&
    \end{aligned}
    \end{equation*}
    \claimname~\ref{claim:closed-inv-homom-helper} is given and proved below.
    \begin{claim}
    \label{claim:closed-inv-homom-helper}
        For all $s \in \trees_\Delta(\Z_{[k]})$ and $q \in Q$, we have that
        \begin{equation*}
            \subst{\lang[\aut_\Z ; q] (s)}{x_{(q', i)} \leftarrow \lang[\aut ; q'](h(t_i)) \mid q' \in Q, i \in [k]} = \lang[\aut; q] (\subst{s}{z_i \leftarrow h(t_i) \mid i \in [k]}) \,.
        \end{equation*}
    \end{claim}
    The proof of the claim is again via structural induction, on $s$:\\
    For $s = \gamma \in \Delta^{(0)}$, we have
    \begin{equation*}
    \begin{aligned}
        &\phantom{{}={}} \subst{\lang[\aut_\Z ; q] (\gamma)}{x_{(q',i)} \leftarrow \lang[\aut; q'](h(t_i)) \mid q' \in Q, i \in [k]} && \\
        &= \lang[\aut; q](\gamma) && \\
        &= \lang[\aut; q] \left( \subst{\gamma}{z_i \leftarrow h(t_i) \mid i \in [k]} \right) \,. &&
    \end{aligned}
    \end{equation*}
    If $s = z_j$ for some $j \in [k]$, we have
    \begin{equation*}
    \begin{aligned}
        &\phantom{{}={}} \subst{\lang[{\aut_\Z ; q}] (z_j)}{x_{(q', i)} \leftarrow \lang[\aut; q'] (h(t_i)) \mid q' \in Q, i \in [k]} && \\
        &= \subst{x_{(q,j)}}{x_{(q', i)} \leftarrow \lang[\aut; q'] (h(t_i)) \mid q' \in Q, i \in [k]} && \\
        &= \lang[\aut; q] (h(t_j)) && \\
        &= \lang[\aut; q] \left( \subst{z_j}{z_i \leftarrow h(t_i) \mid i \in [k]} \right) \,. &&
    \end{aligned}
    \end{equation*}
    Finally, if $s = \gamma[s_1, \dots, s_\ell]$ for some $\ell > 0$, we assume that the claim holds for all $s_i$ and get
    \begin{equation*}
    \begin{aligned}
        &\phantom{{}={}} \subst{\lang[\aut_\Z, q] (\gamma[s_1, \dots, s_\ell])}{x_{(q', i)} \leftarrow \lang[\aut; q'](h(t_i)) \mid q' \in Q, i \in [k]} && \\
        &= \subst{ \left( \subst{\delta_\Z (q, \gamma)}{x_{(q',i)} \leftarrow \lang[\aut_\Z ; q'] (s_i) \mid q' \in Q, i \in [\ell]} \right) }{x_{(q', i)} \leftarrow \lang[\aut; q'](h(t_i)) \mid q' \in Q, i \in [k]} && \\
        &= \subst{ \left( \subst{\delta (q, \gamma)}{x_{(q',i)} \leftarrow \lang[\aut_\Z ; q'] (s_i) \mid q' \in Q, i \in [\ell]} \right) }{x_{(q', i)} \leftarrow \lang[\aut; q'](h(t_i)) \mid q' \in Q, i \in [k]} && \\
        &= \subst{\delta (q, \gamma)}{x_{(q',i)} \leftarrow \left( \subst{\lang[\aut_\Z ; q'] (s_i)} {x_{(p,j)} \leftarrow \lang[\aut; p](h(t_j)) \mid p \in Q, j \in [k]} \right) \mid {q' \in Q, i \in [\ell]}} && \\
        &= \subst{\delta (q, \gamma)}{x_{(q',i)} \leftarrow
        \lang[\aut; q'] \left( \subst{s_i}{z_j \leftarrow h(t_j) \mid j \in [k]} \right)
        \mid {q' \in Q, i \in [\ell]}} \qquad \text{(by assumption)} && \\
        &= \lang[\aut; q] \left( \subst{\gamma[s_1, \dots, s_\ell]}{z_i \leftarrow h(t_i) \mid i \in [k]} \right) \,. &&
    \end{aligned}
    \end{equation*}
    This concludes the proof of \claimname~\ref{claim:closed-inv-homom-helper}, and thereby the proof of \theoremname~\ref{thm:closed-inv-homom}.
\end{proof}

\section{A characterization theorem}
\label{sec:characterization}

In this section, we characterize the class of weighted tree languages defined by wata in terms of wta and tree homomorphisms. We begin by proving a most central decomposition result, which will directly lead to the main theorem of this paper (\theoremname~\ref{thm:characterization}).

\begin{proposition}
\label{prop:decomp}
    $\RecA{\sring} \subseteq \HOM \mathrel{;} \Rec{\sring}$.
\end{proposition}
\begin{proof}
    We show that, given a $(\sig, \sring)$-wata $\aut$, there is a ranked alphabet $\Gamma$, a tree homomorphism $h\colon \trees_\sig \to \trees_\Gamma$, and a $(\Gamma, \sring)$-wta $\aut'$ such that $\lang[\aut] = h \mathrel{;} \lang[\aut']$.
    
    Let $\aut = (\sig, Q, \sring, \pi, \delta)$ be normalized (by \lemmaname s~\ref{lem:normalized} and \ref{lem:degree1}), and let $Q = \SET{q_1, \dots, q_n}$. Since $\aut$ is normalized, in any monomial making up part of a polynomial in $\Ima(\delta)$, any variable $x_{(q,i)}$ occurs (with exponent $1$) at most once, so a subtree may be used at most $n$ times in the computation induced by such a monomial. The idea now is to define $\Gamma$ and $h$ in such a way that the repeated use of a single subtree in these monomial computations can be replaced by the copying capability of $h$.
    
    For each $k \geq 0$, we let $\Gamma^{(n \cdot k)} = \sig^{(k)}$, i.e., each symbol $\sigma$ in $\sig$ of rank $k$ is a symbol of rank $n \cdot k$ in $\Gamma$. We define $h\colon \trees_\sig \to \trees_\Gamma$ by, for each $k \geq 0$ and $\sigma \in \sig^{(k)}$, letting
    \begin{equation*}
        h_k(\sigma) = \sigma[\underbrace{z_1, \dots, z_1}_{n}, \dots, \underbrace{z_k, \dots, z_k}_{n}] \,.
    \end{equation*}
    Intuitively, $h(t)$ is the result of recursively making $n$ copies of each subtree of $t$.

    We now define the wta $\aut' = (\Gamma, Q \cup \SET{q_r}, Q_d, R, \sring, \wt)$, where $q_r \notin Q$. Since $\aut$ is normalized, there is a $Q' \subseteq Q$ such that $\pi = \sum_{q \in Q'} x_q$, and we let $Q_d = Q'$. For each $k \geq 0$, $\sigma \in \Gamma^{(n \cdot k)}$, and $q \in Q$, we define $R_{q;\sigma}$ and $\wt$ as follows. Let $m_1, \dots, m_\ell$ be the unique monomials (of degree at most $1$) whose sum makes up the polynomial $\delta(q,\sigma)$, and let $c_1, \dots, c_\ell$ be their corresponding coefficients. For each $j \in [\ell]$, we let $R_{q;\sigma}$ contain the rule $q[\sigma] \xrightarrow{c_j} \langle p_{(1,1)}, \dots, p_{(n,1)}, \dots, p_{(1,k)}, \dots, p_{(n,k)} \rangle$, where $p_{(i,i')} = q_i$ if $x_{(q_i, i')}$ occurs in $m_j$ and $p_{(i,i')} = q_r$ otherwise. For all $k \geq 0$ and $\sigma \in \Gamma^{(k)}$, $R_{q_r;\sigma}$ contains the single rule $q_r[\sigma] \xrightarrow{1} \langle \underbrace{q_r, \dots, q_r}_{k} \rangle$. Note that $\lang[\aut'; q_r](t) = 1$ for any $t \in \trees_\Gamma$.

    To show $h \mathrel{;} \lang[\aut'] = \lang[\aut]$, we note that, due to the definition of $Q_d$, it is enough to show $\lang[\aut'; q](h(t)) = \lang[\aut; q](t)$ for all $q \in Q$ and $t \in \trees_\sig$, which we do via structural induction on $t = \sigma[t_1, \dots, t_k]$.

    For $k = 0$, we have
    \begin{equation*}
        \lang[\aut'; q](h(\sigma)) = \lang[\aut'; q](\sigma) = \wt(q,\sigma, \langle\rangle) = \delta(q, \sigma) = \lang[\aut; q](\sigma) \,.
    \end{equation*}

    Next, we assume $k > 0$ and that $\lang[\aut'; q'](h(t_i)) = \lang[\aut; q'](t_i)$ holds for all $q' \in Q$ and $i \in [k]$. Let $m_1, \dots, m_\ell$ and $c_1, \dots, c_\ell$ be defined as earlier and, for all $j \in [\ell]$, let $I_j = \SETC{(i, i') \in [n] \times [k]}{x_{(q_i,i')} \text{ occurs in } m_j}$. Then we have
    \begin{equation*}
    \begin{gathered}[b]
    \begin{aligned}
        \lang[\aut'; q](h(\sigma[t_1, \dots, t_k])) &= \lang[\aut'; q] \left( \subst{h_k(\sigma)}{z_i \leftarrow h(t_i) \mid i \in [k]} \right) && \\
        &= \lang[\aut'; q](\sigma[\underbrace{h(t_1), \dots, h(t_1)}_{n}, \dots, \underbrace{h(t_k), \dots, h(t_k)}_{n}]) && \\
        &= \sum_{\substack{\rho = (q,\sigma,w) \in R_{q;\sigma} \\ w = \langle p_{(1,1)}, \dots, p_{(n,1)}, \dots, p_{(1,k)}, \dots, p_{(n,k)} \rangle}} \wt(\rho) \cdot \prod_{(i,i') \in [n] \times [k]} \lang[\aut'; p_{(i,i')}](h(t_{i'})) && \\
        &= \sum_{j \in [\ell]} c_j \cdot \left( \prod_{(i, i') \in I_j} \lang[\aut'; q_i] (h(t_{i'})) \right) \cdot \prod_{(i,i') \in ([n] \times [k]) \setminus I_j} \lang[\aut'; q_r] (h(t_{i'})) && \\
        &= \sum_{j \in [\ell]} c_j \cdot \prod_{(i,i') \in I_j} \lang[\aut'; q_i] (h(t_{i'})) && \\
        &= \sum_{j \in [\ell]} c_j \cdot \prod_{(i,i') \in I_j} \lang[\aut; q_i] (t_{i'}) \qquad \text{(by assumption)} && \\
        &= \sum_{j \in [\ell]} \subst{m_j}{x_{(q_i, i')} \leftarrow \lang[\aut; q_i] (t_{i'}) \mid (i, i') \in [n] \times [k]} && \\
        &= \subst{\delta(q, \sigma)}{x_{(q_i, i')} \leftarrow \lang[\aut; q_i] (t_{i'}) \mid (i, i') \in [n] \times [k]} && \\
        &= \lang[\aut; q] (\sigma[t_1, \dots, t_k]) \,. &&
    \end{aligned}
    \\
    \end{gathered}
    \qedhere
    \end{equation*}
\end{proof}

Using the decomposition of \propositionname~\ref{prop:decomp} together with \theoremname~\ref{thm:closed-inv-homom}, we can quite easily prove that the class of weighted tree languages defined by weighted alternating tree automata is exactly the closure of the recognizable weighted tree languages under inverse tree homomorphisms. This result generalizes \cite[Thm.~$4.5$]{grabolle2023a-nivat-theorem} from strings to trees.
\begin{theorem}
\label{thm:characterization}
    $\RecA{\sring} = \HOM \mathrel{;} \Rec{\sring}$.
\end{theorem}
\begin{proof}
    One inclusion is already given by \propositionname~\ref{prop:decomp}. The other inclusion is derived from \theoremname~\ref{thm:closed-inv-homom} together with the fact that $\Rec{\sring} \subseteq \RecA{\sring}$:
    \begin{equation*}
        \HOM \mathrel{;} \Rec{\sring} \subseteq \HOM \mathrel{;} \RecA{\sring} \subseteq \RecA{\sring} \,. \qedhere
    \end{equation*}
\end{proof}

The characterization in \theoremname~\ref{thm:characterization} can be used to derive certain closure properties of the wata-defined weighted tree languages by leveraging known closure properties of the recognizable weighted tree languages.
We demonstrate this idea on the closure under semiring homomorphisms.
Given two semirings $\sring_1 = (\sringset_1, +_1, \times_1, 0_1, 1_1)$ and $\sring_2 = (\sringset_2, +_2, \times_2, 0_2, 1_2)$, a \emph{semiring homomorphism} is a function $f \colon \sringset_1 \to \sringset_2$ with $f(0_1) = 0_2$, $f(1_1) = 1_2$, $f(a +_1 b) = f(a) +_2 f(b)$, and $f(a \times_1 b) = {f(a) \times_2 f(b)}$ for all $a,b \in \sringset_1$.
A class $C$ of weighted tree languages is \emph{closed under semiring homomorphisms} if, for each $\lang \colon \trees_\sig \to \sring_1$ in $C$, semiring $\sring_2$, and semiring homomorphism $f \colon \sring_1 \to \sring_2$, we have that $\lang \mathrel{;} f \colon \trees_\sig \to \sring_2$ is in $C$.
Below, we will use the fact that $\Rec{\_}$ is closed under semiring homomorphisms~\cite[Lem.~$3$]{borchardt2006cut-sets-as-rec}.

\begin{proposition}
\label{prop:semiring-homom}
    $\RecA{\_}$ is closed under semiring homomorphisms.
\end{proposition}
\begin{proof}
    Let $\sring_1$ and $\sring_2$ be two arbitrary semirings and let $\xhom{(\sring_1, \sring_2)}{}{}$ denote the class of all semiring homomorphisms from $\sring_1$ to $\sring_2$. Then
    \begin{equation*}
    \begin{gathered}[b]
    \begin{aligned}
        \RecA{\sring_1} \mathrel{;} \xhom{(\sring_1, \sring_2)}{}{} &\subseteq \HOM \mathrel{;} \Rec{\sring_1} \mathrel{;} \xhom{(\sring_1, \sring_2)}{}{} \\
        &\subseteq \HOM \mathrel{;} \Rec{\sring_2} \\
        &\subseteq \RecA{\sring_2} \,.
    \end{aligned}
    \\
    \end{gathered}
    \qedhere
    \end{equation*}
\end{proof}

\theoremname~\ref{thm:characterization} also allows us to give an alternative proof of the following known result, generalizing the closure under inverse tree homomorphisms from unweighted recognizable tree languages to recognizable tree languages over locally finite semirings.

\begin{corollary}[{\cite[Thm.~$7.1$]{grabolle2023a-nivat-theorem}}]
    $\HOM \mathrel{;} \Rec{\sring} = \Rec{\sring}$ if and only if $\sring$ is locally finite.
\end{corollary}
\begin{proof}
 By \theoremname~\ref{thm:characterization}, $\HOM \mathrel{;} \Rec{\sring} = \RecA{\sring}$, and by \theoremname~\ref{thm:eq-iff-loc-fin}, $\RecA{\sring} = \Rec{\sring}$ if and only if $\sring$ is locally finite.
\end{proof}

\section{Conclusion}

We have introduced weighted alternating tree automata (wata) over commutative semirings, a common generalization of weighted tree automata (wta), alternating tree automata, and weighted alternating string automata. Their expressive power is equivalent to that of wta if weighted over semirings that are locally finite, and strictly greater if not (\theoremname~\ref{thm:eq-iff-loc-fin}). We have seen that the class of weighted tree languages defined by wata is the closure of the recognizable weighted tree languages under inverse tree homomorphisms (\theoremname~\ref{thm:characterization}). These results give us a clear picture of how wata and wta relate to one another, and provides evidence of a natural fit of wata into the automata landscape.

The theory of tree automata becomes much more complex upon introducing weights, and although there have been many focused efforts (see, e.g., \cite{fulop2022weighted-tree-a}), the weighted setting is still largely underexplored. The study of wata and related concepts can play a part in deepening the understanding of weighted tree languages. This further study could follow one or more of the routes outlined below.

As a first possible investigation, one could take inspiration from the unweighted setting and study the effects of \emph{bounded} alternation. In the context of wata, bounded alternation would correspond to a two-mode wata being restricted to transition between sum and product states at most a fixed number of times. Another option could be to turn to weight structures more general than semirings, such as strong bimonoids or multioperator monoids, and study wata defined over those. Finally, by specifying a suitable grammar model or (fragment of a) logic, one could establish correspondences in the style of Kleene or B\"{u}chi--Elgot--Trakhtenbrot, as has been done for wta~\cite{alexandrakis1987weighted-gramma,droste2006weighted-tree-a}. In addition, such connections can then potentially be leveraged to reason about decidability problems concerning wata, such as zeroness and equality.

\paragraph{Acknowledgments.}
The author wishes to thank Frank Drewes and Marco Kuhlmann for their insightful comments on this paper.
This work was partially supported by the Wallenberg AI, Autonomous Systems and Software Program (WASP) funded by the Knut and Alice Wallenberg Foundation.

%
%
\nocite{*}
\bibliographystyle{eptcs}
\bibliography{refs}

\begin{thebibliography}{10}
\providecommand{\bibitemdeclare}[2]{}
\providecommand{\surnamestart}{}
\providecommand{\surnameend}{}
\providecommand{\urlprefix}{Available at }
\providecommand{\url}[1]{\texttt{#1}}
\providecommand{\href}[2]{\texttt{#2}}
\providecommand{\urlalt}[2]{\href{#1}{#2}}
\providecommand{\doi}[1]{doi:\urlalt{https://doi.org/#1}{#1}}
\providecommand{\eprint}[1]{arXiv:\urlalt{https://arxiv.org/abs/#1}{#1}}
\providecommand{\bibinfo}[2]{#2}

\bibitemdeclare{article}{alexandrakis1987weighted-gramma}
\bibitem{alexandrakis1987weighted-gramma}
\bibinfo{author}{Athanasios \surnamestart Alexandrakis\surnameend} \&
  \bibinfo{author}{Symeon \surnamestart Bozapalidis\surnameend}
  (\bibinfo{year}{1987}): \emph{\bibinfo{title}{Weighted grammars and Kleene's
  theorem}}.
\newblock {\slshape \bibinfo{journal}{Information Processing Letters}}
  \bibinfo{volume}{24}(\bibinfo{number}{1}), pp. \bibinfo{pages}{1--4},
  \doi{10.1016/0020-0190(87)90190-6}.

\bibitemdeclare{inproceedings}{baader2016reasoning-with-}
\bibitem{baader2016reasoning-with-}
\bibinfo{author}{Franz \surnamestart Baader\surnameend} \&
  \bibinfo{author}{Andreas \surnamestart Ecke\surnameend}
  (\bibinfo{year}{2016}): \emph{\bibinfo{title}{Reasoning with Prototypes in
  the Description Logic {$\mathcal{A \mkern-3mu L \mkern-1mu C}$} Using
  Weighted Tree Automata}}.
\newblock In \bibinfo{editor}{Adrian-Horia \surnamestart Dediu\surnameend},
  \bibinfo{editor}{Jan \surnamestart Janou{\v{s}}ek\surnameend},
  \bibinfo{editor}{Carlos \surnamestart Mart{\'\i}n-Vide\surnameend} \&
  \bibinfo{editor}{Bianca \surnamestart Truthe\surnameend}, editors: {\slshape
  \bibinfo{booktitle}{Language and Automata Theory and Applications}},
  \bibinfo{publisher}{Springer International Publishing},
  \bibinfo{address}{Cham}, pp. \bibinfo{pages}{63--75},
  \doi{10.1007/978-3-319-30000-9_5}.

\bibitemdeclare{article}{borchardt2005bounds-for-tree}
\bibitem{borchardt2005bounds-for-tree}
\bibinfo{author}{Bj{\"{o}}rn \surnamestart Borchardt\surnameend},
  \bibinfo{author}{Zolt{\'{a}}n \surnamestart F{\"{u}}l{\"{o}}p\surnameend},
  \bibinfo{author}{Zsolt \surnamestart Gazdag\surnameend} \&
  \bibinfo{author}{Andreas \surnamestart Maletti\surnameend}
  (\bibinfo{year}{2005}): \emph{\bibinfo{title}{Bounds for Tree Automata with
  Polynomial Costs}}.
\newblock {\slshape \bibinfo{journal}{Journal of Automata, Languages and
  Combinatorics}} \bibinfo{volume}{10}(\bibinfo{number}{2--3}), pp.
  \bibinfo{pages}{107--157}, \doi{10.25596/jalc-2005-107}.
\newblock
  \urlprefix\url{https://www.researchgate.net/publication/220520579_Bounds_for_Tree_Automata_with_Polynomial_Costs}.

\bibitemdeclare{article}{borchardt2006cut-sets-as-rec}
\bibitem{borchardt2006cut-sets-as-rec}
\bibinfo{author}{Bj{\"o}rn \surnamestart Borchardt\surnameend},
  \bibinfo{author}{Andreas \surnamestart Maletti\surnameend},
  \bibinfo{author}{Branimir \surnamestart {\v S}e{\v s}elja\surnameend},
  \bibinfo{author}{Andreja \surnamestart Tepav{\v c}evi{\'c}\surnameend} \&
  \bibinfo{author}{Heiko \surnamestart Vogler\surnameend}
  (\bibinfo{year}{2006}): \emph{\bibinfo{title}{Cut sets as recognizable tree
  languages}}.
\newblock {\slshape \bibinfo{journal}{Fuzzy Sets and Systems}}
  \bibinfo{volume}{157}(\bibinfo{number}{11}), pp. \bibinfo{pages}{1560--1571},
  \doi{10.1016/j.fss.2005.11.004}.

\bibitemdeclare{article}{chandra1981alternation}
\bibitem{chandra1981alternation}
\bibinfo{author}{Ashok~K. \surnamestart Chandra\surnameend},
  \bibinfo{author}{Dexter~C. \surnamestart Kozen\surnameend} \&
  \bibinfo{author}{Larry~J. \surnamestart Stockmeyer\surnameend}
  (\bibinfo{year}{1981}): \emph{\bibinfo{title}{Alternation}}.
\newblock {\slshape \bibinfo{journal}{J. ACM}}
  \bibinfo{volume}{28}(\bibinfo{number}{1}), pp. \bibinfo{pages}{114--133},
  \doi{10.1145/322234.322243}.

\bibitemdeclare{inproceedings}{drewes2025dynamically-wei}
\bibitem{drewes2025dynamically-wei}
\bibinfo{author}{Frank \surnamestart Drewes\surnameend}, \bibinfo{author}{Marco
  \surnamestart Kuhlmann\surnameend} \& \bibinfo{author}{Olle \surnamestart
  Torstensson\surnameend} (\bibinfo{year}{2025}):
  \emph{\bibinfo{title}{Dynamically Weighted Tree Transducers}}.
\newblock In \bibinfo{editor}{Giuseppa \surnamestart Castiglione\surnameend} \&
  \bibinfo{editor}{Sabrina \surnamestart Mantaci\surnameend}, editors:
  {\slshape \bibinfo{booktitle}{Implementation and Application of Automata}},
  \bibinfo{publisher}{Springer Nature Switzerland}, \bibinfo{address}{Cham},
  pp. \bibinfo{pages}{115--128}, \doi{10.1007/978-3-032-02602-6_9}.

\bibitemdeclare{incollection}{droste2009semirings-and-f}
\bibitem{droste2009semirings-and-f}
\bibinfo{author}{Manfred \surnamestart Droste\surnameend} \&
  \bibinfo{author}{Werner \surnamestart Kuich\surnameend}
  (\bibinfo{year}{2009}): \emph{\bibinfo{title}{Semirings and Formal Power
  Series}}.
\newblock In \bibinfo{editor}{Manfred \surnamestart Droste\surnameend},
  \bibinfo{editor}{Werner \surnamestart Kuich\surnameend} \&
  \bibinfo{editor}{Heiko \surnamestart Vogler\surnameend}, editors: {\slshape
  \bibinfo{booktitle}{Handbook of Weighted Automata}},
  \bibinfo{publisher}{Springer Berlin Heidelberg}, \bibinfo{address}{Berlin,
  Heidelberg}, pp. \bibinfo{pages}{3--28}, \doi{10.1007/978-3-642-01492-5_1}.

\bibitemdeclare{article}{droste2006weighted-tree-a}
\bibitem{droste2006weighted-tree-a}
\bibinfo{author}{Manfred \surnamestart Droste\surnameend} \&
  \bibinfo{author}{Heiko \surnamestart Vogler\surnameend}
  (\bibinfo{year}{2006}): \emph{\bibinfo{title}{Weighted tree automata and
  weighted logics}}.
\newblock {\slshape \bibinfo{journal}{Theoretical Computer Science}}
  \bibinfo{volume}{366}(\bibinfo{number}{3}), pp. \bibinfo{pages}{228--247},
  \doi{10.1016/j.tcs.2006.08.025}.

\bibitemdeclare{article}{fulop2011weighted-extend}
\bibitem{fulop2011weighted-extend}
\bibinfo{author}{Zolt{\'{a}}n \surnamestart F{\"u}l{\"o}p\surnameend},
  \bibinfo{author}{Andreas \surnamestart Maletti\surnameend} \&
  \bibinfo{author}{Heiko \surnamestart Vogler\surnameend}
  (\bibinfo{year}{2011}): \emph{\bibinfo{title}{Weighted Extended Tree
  Transducers}}.
\newblock {\slshape \bibinfo{journal}{Fundam. Inform.}} \bibinfo{volume}{111},
  pp. \bibinfo{pages}{163--202}, \doi{10.3233/FI-2011-559}.

\bibitemdeclare{article}{fulop2022weighted-tree-a}
\bibitem{fulop2022weighted-tree-a}
\bibinfo{author}{Zolt{\'{a}}n \surnamestart F{\"{u}}l{\"{o}}p\surnameend} \&
  \bibinfo{author}{Heiko \surnamestart Vogler\surnameend}
  (\bibinfo{year}{2026}): \emph{\bibinfo{title}{Weighted Tree Automata -- May
  it be a little more?}}
\newblock {\slshape \bibinfo{journal}{CoRR}}
  \bibinfo{volume}{abs/2212.05529v3}, \doi{10.48550/ARXIV.2212.05529}.

\bibitemdeclare{article}{ghorani2016alternating-reg}
\bibitem{ghorani2016alternating-reg}
\bibinfo{author}{Maryam \surnamestart Ghorani\surnameend} \&
  \bibinfo{author}{Mohammad~Mehdi \surnamestart Zahedi\surnameend}
  (\bibinfo{year}{2016}): \emph{\bibinfo{title}{Alternating Regular Tree
  Grammars in the Framework of Lattice-Valued Logic}}.
\newblock {\slshape \bibinfo{journal}{Iranian Journal of Fuzzy Systems}}
  \bibinfo{volume}{13}(\bibinfo{number}{2}), pp. \bibinfo{pages}{71--94},
  \doi{10.22111/ijfs.2016.2360}.

\bibitemdeclare{article}{grabolle2023a-nivat-theorem}
\bibitem{grabolle2023a-nivat-theorem}
\bibinfo{author}{Gustav \surnamestart Grabolle\surnameend}
  (\bibinfo{year}{2023}): \emph{\bibinfo{title}{A Nivat Theorem for Weighted
  Alternating Automata over Commutative Semirings}}.
\newblock {\slshape \bibinfo{journal}{Logical Methods in Computer Science}}
  \bibinfo{volume}{Volume 19, Issue 4}:\bibinfo{eid}{27},
  \doi{10.46298/lmcs-19(4:27)2023}.

\bibitemdeclare{article}{kostolanyi2018alternating-wei}
\bibitem{kostolanyi2018alternating-wei}
\bibinfo{author}{Peter \surnamestart Kostol{\'a}nyi\surnameend} \&
  \bibinfo{author}{Filip \surnamestart Mi{\v s}{\'u}n\surnameend}
  (\bibinfo{year}{2018}): \emph{\bibinfo{title}{Alternating weighted automata
  over commutative semirings}}.
\newblock {\slshape \bibinfo{journal}{Theoretical Computer Science}}
  \bibinfo{volume}{740}, pp. \bibinfo{pages}{1--27},
  \doi{10.1016/j.tcs.2018.05.003}.

\bibitemdeclare{article}{seidl1994finite-tree-aut}
\bibitem{seidl1994finite-tree-aut}
\bibinfo{author}{Helmut \surnamestart Seidl\surnameend} (\bibinfo{year}{1994}):
  \emph{\bibinfo{title}{Finite tree automata with cost functions}}.
\newblock {\slshape \bibinfo{journal}{Theoretical Computer Science}}
  \bibinfo{volume}{126}(\bibinfo{number}{1}), pp. \bibinfo{pages}{113--142},
  \doi{10.1016/0304-3975(94)90271-2}.

\bibitemdeclare{article}{slutzki1985alternating-tre}
\bibitem{slutzki1985alternating-tre}
\bibinfo{author}{Giora \surnamestart Slutzki\surnameend}
  (\bibinfo{year}{1985}): \emph{\bibinfo{title}{Alternating tree automata}}.
\newblock {\slshape \bibinfo{journal}{Theoretical Computer Science}}
  \bibinfo{volume}{41}, pp. \bibinfo{pages}{305--318},
  \doi{10.1016/0304-3975(85)90077-5}.

\end{thebibliography}

\end{document}